%% file: main.tex
\documentclass[11pt]{article}

\usepackage[T5]{fontenc}
\usepackage[margin=1in]{geometry}
\usepackage[labelfont=bf]{caption}
\usepackage{tabularray}
\usepackage{url}
\usepackage{algorithm,algpseudocode,amsmath,amssymb,amsthm,float,mathtools,thmtools,thm-restate,scalerel,stackengine,verbatim,subcaption,bbm}
\usepackage[shortlabels]{enumitem}
\usepackage{tikz}
\usetikzlibrary{decorations.pathreplacing, positioning}
 
\stackMath
\newcommand\reallywidehat[1]{%
\savestack{\tmpbox}{\stretchto{%
  \scaleto{%
      \scalerel*[\widthof{\ensuremath{#1}}]{\kern-.6pt\bigwedge\kern-.6pt}%
          {\rule[-\textheight/2]{1ex}{\textheight}}%WIDTH-LIMITED BIG WEDGE
            }{\textheight}% 
            }{0.5ex}}%
            \stackon[1pt]{#1}{\tmpbox}%
            }
\newcommand{\bool}{\set{0,1}}

\newcommand{\Nbb}{\mathbb{N}}
\newcommand{\Rbb}{\mathbb{R}}

\newcommand{\rb}{\mathbf{r}}
\newcommand{\tb}{\mathbf{t}}

\newcommand{\Ac}{\mathcal{A}}
\newcommand{\Bc}{\mathcal{B}}
\newcommand{\Cc}{\mathcal{C}}
\newcommand{\Dc}{\mathcal{D}}
\newcommand{\Ec}{\mathcal{E}}

\newcommand{\Sc}{\mathcal{S}}

\newcommand{\wt}[1]{\widetilde{#1}}

\newcommand{\eps}{\varepsilon}

\newcommand{\bO}[1]{\operatorname*{O}\paren*{#1}}
\newcommand{\bOt}[1]{\operatorname*{\wt{O}}\paren*{#1}}

\newcommand{\bOm}[1]{\operatorname*{\Omega}\paren*{#1}}

\newcommand{\bT}[1]{\operatorname*{\Theta}\paren*{#1}}

\newcommand{\plog}{\operatorname*{\normalfont{\text{polylog}}}}

\DeclareMathOperator{\poly}{poly}

\DeclarePairedDelimiter\floor{\lfloor}{\rfloor}
\DeclarePairedDelimiter\ceil{\lceil}{\rceil}

\DeclarePairedDelimiter\abs{|}{|}
\DeclarePairedDelimiter\brac{\lbrack}{\rbrack}
\DeclarePairedDelimiter\set{\lbrace}{\rbrace}
\DeclarePairedDelimiter\paren{\lparen}{\rparen}
\DeclarePairedDelimiter\interval{\lbrack}{\rparen}

\DeclarePairedDelimiter\bra{\langle}{|}
\DeclarePairedDelimiter\ket{|}{\rangle}
\DeclarePairedDelimiterX\braket[2]{\langle}{\rangle}{#1\delimsize\vert#2}

\newtheorem{theorem}{Theorem}
\newtheorem{lemma}[theorem]{Lemma}

\theoremstyle{definition}
\newtheorem{definition}[theorem]{Definition}
\theoremstyle{remark}

\newcommand{\nth}{^\text{th}}

\newcommand{\din}{d^\text{in}}
\newcommand{\dout}{d^\text{out}}
\newcommand{\dedge}[1]{\overrightarrow{#1}}

\newcommand{\hist}[1]{\mathsf{Snap}^{#1}}
\newcommand{\histps}[1]{\mathsf{PsSnap}^{#1}}
\newcommand{\db}[1]{d^{\leq#1}}
\newcommand{\dbps}[1]{\wt{d}^{\leq#1}}
\newcommand{\da}[1]{d^{>#1}}
\newcommand{\doutb}[1]{d^{\text{out},\leq#1}}
\newcommand{\doutbps}[1]{\wt{d}^{\text{out},\leq#1}}
\newcommand{\douta}[1]{d^{\text{out},>#1}}
\newcommand{\bps}[1]{\wt{b}^{#1}}

\newcommand{\inc}{\mathtt{inc}}
\newcommand{\measure}{\mathtt{measure}}
\newcommand{\cleanup}{\mathtt{cleanup}}

\newcommand{\accuracy}{\kappa}
\newcommand{\prefix}{\rho}

\newcommand{\mcut}{{\normalfont\textsc{Max-Cut}}}
\newcommand{\mdicut}{{\normalfont\textsc{Max-DiCut}}}
\newcommand{\mdcut}{\mdicut}
\newcommand{\mtwoAnd}{{\normalfont\textsc{Max-2AND}}}

{\makeatletter
	\gdef\commentmark{%
		\expandafter\ifx\csname @mpargs\endcsname\relax % in minipage?
		\expandafter\ifx\csname @captype\endcsname\relax % in figure/caption?
		\marginpar{comment}% not in a caption or minipage, can use marginpar
		\else
		comment % notice trailing space
		\fi
		\else
		comment % notice trailing space-
		\fi}
	\gdef\comment{\@ifnextchar[\comment@lab\comment@nolab}
	\long\gdef\comment@lab[#1]#2{{\bf [\commentmark #2 ---{\sc #1}]}}
	\long\gdef\comment@nolab#1{{\bf [\commentmark #1]}}
}

\title{Exponential Quantum Space Advantage for Approximating Maximum Directed
Cut in the Streaming Model}
\date{}
\author{John Kallaugher\\Sandia National
Laboratories\\\texttt{jmkall@sandia.gov} \and Ojas Parekh\\Sandia National
Laboratories\\\texttt{odparek@sandia.gov} \and Nadezhda
Voronova\\ Boston University\\\texttt{voronova@bu.edu}}

\begin{document}
\maketitle

\begin{abstract}
\noindent
While the search for quantum advantage typically focuses on speedups in
execution time, quantum algorithms also offer the potential for advantage in
\emph{space} complexity. Previous work has shown such advantages for data
stream problems, in which elements arrive and must be processed sequentially
without random access, but these have been restricted to specially-constructed
problems \lbrack Le Gall, SPAA `06\rbrack{} or polynomial advantage
\lbrack Kallaugher, FOCS `21\rbrack.  We show an \emph{exponential} quantum
space advantage for the maximum directed cut problem.  This is the first known
exponential quantum space advantage for any natural streaming problem. This
also constitutes the first unconditional exponential quantum resource advantage
for approximating a discrete optimization problem in any setting.

Our quantum streaming algorithm $0.4844$-approximates the value of the largest
directed cut in a graph stream with $n$ vertices using $\plog(n)$ space, while
previous work by Chou, Golovnev, and Velusamy \lbrack FOCS '20\rbrack{} implies
that obtaining an approximation ratio better than $4/9 \approx 0.4444$ requires
$\bOm{\sqrt{n}}$ space for any classical streaming algorithm. Our result is
based on a recent $\bOt{\sqrt{n}}$ space classical streaming approach by
Saxena, Singer, Sudan, and Velusamy \lbrack FOCS '23\rbrack, with an additional
improvement in the approximation ratio due to recent work by Singer \lbrack
APPROX '23\rbrack. 
\end{abstract}

\input{introduction}
\input{proof_overview}
\input{preliminaries}
\input{dicut_to_bias}

\input{bias_to_pseudobias}
\input{pseudobias_quantum_algorithm}
\input{max_dicut_quantum_algorithm}
\input{acknowledgements}

\newpage
\bibliographystyle{alpha}
\bibliography{refs}

\end{document}

%% file: introduction.tex
\section{Introduction}
Streaming algorithms are a means of processing very large data sets, in
particular those too large to be stored wholly in memory.  In this setting data
elements arrive sequentially, and a streaming algorithm must process elements
as they arrive using as little space as possible---ideally logarithmic in the
size of the data.  Streaming algorithms have been developed for many
applications, including computing statistics of data streams and estimating
graph parameters~\cite{FM85,AMS96,BKS02}. The practical use of such algorithms
goes back to the approximate counting algorithm of Morris~\cite{Mor78,Fla85}
and has expanded as the growth of the internet has increased the prevalence of
extremely large datasets~\cite{RLSLT13}.

Quantum computing offers the prospect of exponential resource advantages over
conventional classical computers.  The primary resource of interest has
traditionally been execution time, and while a handful of examples such as
Shor's celebrated algorithm for integer factorization~\cite{S99} provide
exponential speedups over the best-known classical counterparts, provable
exponential speedups over the best-possible classical algorithms remain largely
elusive. Although speedups are important, space is an especially critical
resource for quantum computing, as scalable fault-tolerant qubits are and will
likely continue to be scarce.  Quantum streaming algorithms offer a natural
avenue for exploring space-efficient quantum algorithms.        

Moreover, space-efficient quantum algorithms, including streaming algorithms,
offer an alternative opportunity for quantum advantage. Very large datasets
continue to be prevalent in computing, and so an algorithm with quantum memory
can potentially process much larger datasets than one with only classical
memory, provided the problem in question evinces a large enough quantum
advantage to justify the much higher cost of qubits relative to classical bits.

Provable exponential space advantages for quantum algorithms in the streaming
setting have been known since the seminal work of Gavinsky, Kempe, Kerenidis,
Raz, and de Wolf~\cite{GKKRW07}; however, the problem studied was constructed
for the purpose of proving this separation, leaving open the question of
whether such advantages exist for problems of independent classical interest.
The question of quantum advantage for a ``natural'' streaming problem was
suggested by Jain and Nayak~\cite{JN14}, who proposed a candidate problem of
recognizing the Dyck(2) formal language, related to strings of balanced
parentheses.  It remains open whether quantum advantage is possible for this
problem~\cite{NT17}.  

This question of a quantum advantage for a natural streaming problem was
recently resolved in~\cite{K21}, which demonstrates quantum advantage for the
problem of counting triangles in graph streams.  This problem has been long
studied classically and drives numerous applications in analyzing social
networks~\cite{BKS02,BOV13,JK21}. However, the advantage offered is only
\emph{polynomial} in the input size, and requires additional parametrization of
the input (the latter being unavoidable with the triangle counting problem).
We give the first \emph{exponential} quantum space advantage for a streaming
problem of independent interest.

\paragraph{Our Results} We give a $\plog(n)$-space algorithm that gives a
$0.4844$-approximation for the maximum directed cut (\mdcut) problem. Given a
directed graph $G = (V,E)$ (as a sequence of edges), $\mdcut(G)$ is the
greatest number of edges that can be simultaneously ``cut'' by a partition $V =
V_0 \sqcup V_1$, where $\dedge{uv}$ is cut iff $u \in V_0, v\in V_1$. In the
classical streaming setting, Chou, Golovnev, and Velusamy~\cite{CGV20} showed that any
approximation better than $4/9$ requires $\bOm{\sqrt{n}}$ space.

\begin{restatable}{theorem}{improvedmdcutalg}
\label{cor:improved_mdcutalg}
There is a quantum streaming algorithm which $0.4844$-approximates the $\mdcut$
value of an input graph $G$ with probability $1 - \delta$. The algorithm uses
$\bO{\log^5 n \log \frac{1}{\delta}}$ qubits of space.
\end{restatable}

Our quantum algorithm is inspired by the classical algorithm of Saxena, Singer,
Sudan, and Velusamy~\cite{SSSV23b}, which achieves this same approximation
ratio in a classically-optimal $\bOt{\sqrt{n}}$ space. Like theirs, our
algorithm estimates a histogram of the edges of $G$ where the buckets
correspond to edges between ``bias classes'', which partition the vertices of
$G$ according to their bias (the difference between the in-degree and
out-degree of a vertex, normalized by the degree).  In their algorithm, the
histogram is then used to estimate the output value of a 0.4835-approximation
algorithm for $\mdcut$ due to Feige and Jozeph~\cite{FJ15} that randomly
assigns each vertex to a side of a cut based solely on its bias.  Their
algorithm prescribes a constant number of bias ranges, where the random
assignment only depends on these ranges.  This allows the aforementioned
histogram estimate to be constructed with respect to a constant number of bias
classes. We use a result from recent work\footnote{The result in~\cite{S23}
also generalizes the result from \mdicut{} to \mtwoAnd{}, where the stream
consists of arbitrary two-variable clauses and the goal is to approximate the
maximum number of them that can be simultaneously satisfied. We believe our
algorithm should generalize similarly, but we do not consider the question in
this paper.} by Singer~\cite{S23} that improves the approximation ratio offered
by~\cite{FJ15} with a new set of bias classes, allowing us an improved
approximation ratio (we note that this later work also improves the
approximation ratio offered by the aforementioned classical algorithm; the
advantage of our quantum algorithm lies in its exponentially better space
complexity).

The existence of quantum advantage for this problem contrasts with recent work
of Kallaugher and Parekh~\cite{KP22}, who show that approximating the
\emph{undirected} maximum cut problem (\mcut) in graph streams does not admit
any asymptotic quantum advantage. Random assignment of vertices to a side of a
cut yields a $1/2$-approximation for $\mcut$ and a $1/4$-approximation for
$\mdcut$.  These can be implemented as streaming approximations by counting the
number of edges and dividing by 2 or 4.  While $\mcut$ requires linear space
(whether in the quantum or classical setting) to do better than a
$1/2$-approximation, $\mdcut$ does admit a $\bO{\log n}$-space classical
algorithm that beats the trivial $1/4$-approximation~\cite{CGV20}.  As noted
above, for $\mdcut$, $\plog(n)$-space classical algorithms cannot offer better
than a $4/9$-approximation, while we show such a quantum streaming algorithm
exists.

\paragraph{Quantum Approximation Advantages} Finding provable quantum
advantages for approximating discrete optimization problems is an open problem
that has received considerable attention, especially following the introduction
of the Quantum Approximate Optimization Algorithm
(QAOA)~\cite{farhi2014quantum}.  As canonical examples of constraint
satisfaction problems (CSPs), $\mcut$ and related problems have served as focal
points in QAOA analysis and empirical performance studies.  The approximability
of $\mcut$ and other CSPs is well understood conditional on the unique games
conjecture (UGC)~\cite{raghavendra2008optimal}.  If the latter holds, then for
every CSP there is some $\alpha \in (0,1]$ for which an $\alpha$-approximation
is achievable in polynomial time, but for which it is NP-hard to obtain an
$(\alpha + \varepsilon)$-approximation for any $\varepsilon > 0$.  This leaves
little hope for worst-case quantum approximation advantages, as we do not
expect polynomial time quantum algorithms to solve NP-hard problems. 

Our work, on the other hand, shows that a provable quantum approximation
advantage \emph{is} possible in the space-constrained streaming setting. The
$\plog$-space streaming setting does not admit $\alpha$-approximations of the
form described above, which are based on solving semidefinite programs, leaving more room for quantum advantage. For $\mdcut$,
$\alpha > 0.874$ is known to be possible in general~\cite{LLZ02,BHPZ23}, but the
work of~\cite{CGV20} shows that $4/9$ is the best classically possible for
$\plog$-space streaming algorithms. As previously mentioned, it is impossible for a $\plog$-space
quantum algorithm to attain better than a $1/2$-approximation for any
$\varepsilon > 0$ in the streaming setting.  This is also true for $\mdcut$,
since approximating $\mcut$ can be reduced to approximating $\mdcut$ in
instances where $(v,u)$ is an edge whenever $(u,v)$ is.  Therefore, there is an 
opportunity for $\plog$-space quantum $\beta$-approximations for $\beta \in (4/9,1/2]$,
and we indeed demonstrate that an exponential quantum space advantage is possible in this range.

%% file: proof_overview.tex
\section{Our Techniques}
We follow the approach of~\cite{SSSV23b}, who give an $\bOt{\sqrt{n}}$ upper
bound for $0.4835$-approximation of \mdcut. To achieve this result, they use
the notion of a ``first-order snapshot'' of a graph. Introduced
in~\cite{SSSV23a}, this is a histogram of the frequency with which the
(directed) edges of the graph go from one ``bias class'' to another.

The bias of a vertex $v$ is defined as $b_v = \frac{\dout-\din}{d_v}$ where
$d_v$ denotes the degree of $v$, $\dout_v$ denotes the out-degree of $v$, and
$\din_v$ denotes the out-degree of $v$. Note that $b_v \in \brac{-1,1}$. Given
a partition of $\brac{-1,1}$ into intervals
$\interval{a,b}$ (along with one of the form $\brac{a,1}$, so that all of
$\brac{-1,1}$ is covered), the \emph{bias classes} $H_i$ are the sets of of
vertices whose biases belong to each interval. The first-order snapshot of the
graph is then given by the count of edges $\dedge{uv}$ such that $(u,v) \in H_i
\times H_j$ for each $(i,j)$. From hereon we will drop the ``first-order''
qualification and refer to this simply as a ``snapshot''.

The key tool used in~\cite{SSSV23b} is a result of~\cite{FJ15} stating that
$0.4835$-approximation of the \mdcut{} value of a graph can be computed from
its snapshot, where the set of bias classes defining the snapshot is the same
for all graphs and so, in particular, the number of different bias classes is
constant. We use the result of~\cite{S23} instead, but the central problem is
unchanged: given a constant set of bias thresholds, estimate the corresponding
snapshot. See a formal description in Section \ref{section:dicut_to_bias}. The
authors give a $\bOt{\sqrt{n}}$-space classical algorithm for estimating these
snapshots in the stream, and thereby for computing a $0.4835$-approximation of
the \mdcut{} value of a graph.

We show that the snapshot of a graph (and thus an approximation to its
\mdcut{} value) can be estimated by a quantum streaming algorithm in $\plog(n)$
space. In the remainder of this section, we will describe the main technical
ideas behind this contribution.

\subsection{One-Way Communication, Simplified Problem}
We will start by considering a simplified version of this problem in the
one-way communication setting. Suppose there are two parties, Alice and Bob,
each with their own input. Alice's input is $n$ labelled graph vertices, and
Bob's input is a ``directed matching'' (a set of vertex-disjoint directed
edges), and a pair of labels $i, j$. Alice is allowed to send a message to Bob,
and after receiving this message, Bob's goal is to estimate the number of edges
from his matching that have a vertex labeled $i$ as its head and vertex labeled
$j$ as its tail.  The question in this setting is how small Alice's message
could be.

Note that if Alice chooses $\bO{\sqrt{n}/\eps^2}$ vertices uniformly at random
and sends the sampled vertices along with the labels, one can show (the
``Birthday paradox'') that Bob would be able to estimate (up to $\bO{\eps n}$
additive error) the number of correct edges with probability 2/3. This is (up
to a log factor) optimal for classical protocols by~\cite{GKKRW07}.

On the other hand, with quantum communication, we can achieve a significant
improvement using a slightly modified version of the quantum Boolean Hidden
Matching protocol from~\cite{GKKRW07}. Alice sends to Bob $k$ copies of the
superposition \[
\frac{1}{\sqrt{2n}} \paren*{\sum_{v \in V} \ket{vi_v h} + \ket{vi_vt}}
\]
where $i_v$ is the label of the vertex $v$, and the last register $h, t$
denotes if the element is to be treated as head or tail of an edge. Bob then
measures each copy of the state with the projectors onto the following vectors
along with the projector onto the complement of the space they span. \[
\frac{\ket{uih} \pm \ket{vjt}}{\sqrt{2}}
\]
for each $\dedge{uv}$ from Bob's graph.

For each copy of the superposition sent and measured, each state of the form
$\frac{\ket{uih} + \ket{vjt}}{\sqrt{2}}$ for an edge $\dedge{uv}$ will be
returned with probability $1/n$ if $u$ and $v$ are labelled with $i$ and $j$,
with probability $1/4n$ if the vertices have exactly one correct label, and
probability zero otherwise. If Bob sees such an edge, he adds $n/k$ to his
estimate of the number of edges with head labelled $i$ and tail labelled $j$ 
where $k$ is an accuracy parameter to be set later. 

Each state of the form $\frac{\ket{uih} - \ket{vjt}}{\sqrt{2}}$ will be
returned with probability $1/4n$ if the vertices have exactly one correct
label, and probability zero otherwise. If Bob sees such an edge, he subtracts
$n/k$ from his estimate. 

The expectation of Bob's estimate will therefore be correct, and its variance
will be $\bO{n^2/k}$. If $k = \bT{1/\eps^2}$, Bob's estimate is $\eps n$-close
to the correct value with probability $2/3$. This protocol only uses
$O(\frac{1}{\eps^2} \log{n})$ qubits.

\subsection{One-way Communication, Snapshot Approximation}
Now, let us return to the original problem of estimating the entries of the
graph snapshot, while remaining in the two player one-way communication
setting. Both Alice and Bob are given a directed graph, and their goal is to
estimate the bias histogram of the combined graph.

How is this problem different from the one we considered before? Firstly, Bob's
input is no longer a matching, which means the projectors Bob used are no
longer guaranteed to be orthogonal. We could address this by splitting Bob's
graph into matchings and measuring a different copy of Alice's state with each,
but this would require Alice to send $\bT{d_{\text{max}}}$ copies (where
$d_{\text{max}}$ is the maximum degree of the graph), which eliminates the
advantage we achieved previously\footnote{This same issue arises
in~\cite{K21}. There it was solved by using a classical algorithm to handle
cases where $d_{\text{max}}$ is particularly bad, at the cost of reducing the
exponential advantage to only polynomial.}.

Secondly, we still have to estimate the number of edges between vertices with a
pair of labels, but now the labels (bias classes) depend on the graph itself.
Moreover, neither Alice nor Bob alone know the labels since the labels depend
on both Alice and Bob's input graphs.

Fortunately, we can tackle both problems using properties of how the biases
depend on the two players' graphs. Our first observation is that, if we're
given the biases of a vertex in Alice's input and Bob's input separately, as
well as the degrees of this vertex in each of the graphs, we can obtain its
bias with respect to the whole graph. The second observation is that if the
degree of a vertex in Bob's graph is much higher than its degree in Alice's
graph, Alice's graph contributes very little to the bias of this vertex and
thus Bob can compute an estimate of the bias by himself. 

The former means that, if Alice sketches her vertices in $\plog(n)$ different
subsets based on their degrees and biases (so that for each sketch, Bob knows
the degree and bias of the vertices he is dealing with up to a small error),
Bob has enough information to approximate the biases in the full graph.  The
latter fact means that, if Alice copies each of her vertices with multiplicity
equal to a sufficiently large constant times its degree (which she can afford
to do, as it results in a superposition with no more than $\bO{m}$ states,
where $m$ is the number of edges in her input), either Bob will be able to
measure with all of his edges, or Bob does not need Alice's input to determine
the bias of this vertex.

Suppose Alice sends the superposition
\[
\frac{1}{\sqrt{\abs{H} + \abs{T}}}\sum_{v \in H}\sum_{i \in
\brac{\ceil{d_H/\eps}}} \ket{vih} + \sum_{v \in T}\sum_{i \in
\brac{\ceil{d_T/\eps}}} \ket{vit}
\]
where $H, T$ are subsets of Alice's vertices with biases and degrees in a small
enough range (such that Bob can know the biases and degrees to an $\eps$
approximation), with $d_H$, $d_T$ upper bounds on the degrees of vertices in
$H$ and $T$. Bob may then measure the state with the projectors onto the
vectors
\[
\ket{ui_{\dedge{uv}} t} \pm \ket{vj_{\dedge{uv}} t}
\] 
for each $\dedge{uv}$ in Bob's graph that has the degree of $u$ in Bob's graph
at most $d_H/\eps$ and the degree of $v$ in Bob's graph at most $d_T/\eps$. The notation
$i_{\dedge{uv}}$ denotes the index of edge $\dedge{uv}$ in a fixed ordering of out-edges
of $u$, and $j_{\dedge{uv}}$ denotes the index of edge $\dedge{uv}$ in a fixed ordering
of in-edges of $v$. Now the measurement operators are orthogonal again.

Similarly to the previous case, $\ket{ui_{\dedge{uv}}t} +
\ket{vj_{\dedge{uv}}t}$ will be the measurement outcome with probability
$2/(\abs{H} + \abs{T})$ if $u \in H$ and $v \in T$, and $\ket{ui_{\dedge{uv}}t}
\pm \ket{vi_{\dedge{uv}}t}$ are equally likely otherwise. So Bob can use this
measurement to estimate how many of his edges go from $H$ to $T$, and also to
\emph{sample} from such edges\footnote{With the caveat that he will also sample
pairs that don't exist.  But because each fictitious pair is seen equally often
with a $+$ or $-$ sign, he can still sample ``in expectation''}.

Given such a sampled edge, Bob can calculate the bias of its endpoints using
his knowledge of the degrees and biases of $u$ and $v$ restricted to Alice's
graph, and their degrees and biases in his own graph. So by repeating this
process for $\plog(n)$ many combinations of ranges, the players can approximate
the number of edges between each pair of ``bias classes'', with the exception
of vertices whose degree is much higher (more than $1/\eps$ times as high) for
Bob than Alice. 

For these vertices, Bob can approximate their contribution to the snapshot with
only his own information, as Alice's input only provides an $\eps$ contribution
to their biases. Of course, he does not actually know which vertices these
should be. However, this can be solved by having Bob calculate this for
\emph{every} vertex, and then performing appropriate corrections when the
protocol samples non low-degree vertices.

\subsection{Streaming Algorithm}
Now we want to implement this protocol in the stream. Our first obstacle is
that now, each time a new edge arrives, we need to ``process'' it both as Alice
and as Bob, meaning we need to update the superposition and measure it.
Secondly, instead of performing our measurements all at once, we will need to
perform them edge by edge, at each point using projectors onto a small part of
the space along with a ``complementary projector'' onto the rest of the space.
When a measurement returns something \emph{other} than the complementary
projector, our quantum stage will terminate and we will proceed classically,
but when the complementary projector is returned it will be important to ensure
that the state retains the properties needed for future measurements.  Finally,
we may no longer simply restrict our input to vertices in a specific class of
degrees or biases, as at the time we see an edge we do not know what edges have
arrived incident to its endpoints in the past (let alone the future). So we
need to maintain a quantum state that somehow encodes the degree and bias of
each vertex, in a way such that whenever an edge arrives, we can measure with
it and have zero (expected) effect on our output when the edge's endpoints have
the wrong degrees or biases.

Start by considering the case where we are only interested in whether the
endpoints of the edge have the right \emph{degree}. To make things even
simpler, imagine we only care whether a vertex has degree \emph{at least} $d$.
We may store the state \[
\frac{1}{\sqrt{M'}}\paren*{\Sc + \sum_{v \in V} \sum_{i = 1}^{\brac{d_v}}
\paren*{\ket{vih} + \ket{vit}}}
\] where $d_v$ is the degree of $v$ among the edges that have arrived up until
now, and $\Sc$ is is a collection of ``scratch states'' that do not contain any
information about the graph and are there to be swapped with states we want, in
order to maintain this superposition. When a new edge $\dedge{uv}$ arrives, we
can maintain this state by sending $\ket{uih}\to\ket{u(i+1)h}$,
$\ket{uit}\to\ket{u(i+1)t}$,$\ket{vih}\to\ket{v(i+1)h}$,$\ket{vit}\to\ket{v(i+1)t}$
and swapping\footnote{That is, executing the unitary operation that swaps a
given basis element $\ket{is}$ in the superposition $\Sc = \sum_i \ket{is}$ for
the desired state while leaving the rest of the space unchanged. Note that this
is distinct from the ``swap'' operation that e.g.\ swaps the two registers of a
two-qubit state. By maintaining a single counter we may track how many of these
``scratch states'' have already been used and thus ensure our swap always
targets a state present in the superposition.} states from $\Sc$ for
$\ket{u1h}, \ket{u1t}, \ket{v1h}, \ket{v1t}$. $\Sc$ will therefore require $2m$
states to start (where $m$ is the number of edges in the stream) and so the
normalization factor $M'$ will start at $M = 2m$, before changing as the state
is measured. 

To use this state, after $\dedge{uv}$ arrives and the state has been updated
accordingly, we may measure with the projectors (along with a complementary
projector onto the space not spanned by them) onto the vectors \[
\ket{ud_Hh}\pm\ket{vd_Tt}
\]
where $d_H, d_T$ are the minimum desired degrees for the head and the
tail respectively. Then, if both endpoints have achieved the desired degree,
$\ket{ud_Hh} + \ket{vd_Tt}$ is a possible outcome for the measurement but not
$\ket{ud_Hh} - \ket{vd_Tt}$. Otherwise both are equally likely, and so we can
estimate the number of edges between pairs of vertices with the right degrees
(similar to the 2-player labeling problem). Note that if
$\ket{ud_Hh}\pm\ket{vd_Tt}$ \emph{is} returned, the state collapses to the
corresponding vector---at this point we stop using the state and proceed
classically with the returned vertices $u$ and $v$ (e.g. counting how many in-
and out-edges are seen incident to them after this point). Therefore, each copy
of the quantum state maintained allows us to ``sample'' up to one pair of
vertices $(u,v)$.

One side-effect of this measurement, that will be convenient later, is that
conditioned on returning the complementary projector (i.e. not sampling some
$(u,v)$ and terminating the quantum stage) it ``deletes'' $\ket{ud_Hh}$ and
$\ket{vd_Tt}$ from the superposition after each measurement, so the state we
maintain ends up being \[
\frac{1}{\sqrt{M'}}\paren*{\Sc + \sum_{v \in V} \sum_{i = 1}^{\min(d_v, d_H-1)}
\ket{vih} + \sum_{i = 1}^{\min(d_v, d_T-1)}\ket{vit}}\text{.}
\]
This also means that each measurement before termination reduces $M'$, and so
the probability of returning a given measurement outcome does not depend on how
many edges have been processed so far (as the probability that the algorithm
terminates before a given edge is processed exactly cancels out the decrease in
$M'$ conditioned on the algorithm \emph{not} terminating before then). We can
extend this method to only count (in expectation) edges where $u$ and $v$ are
in \emph{ranges} $\interval{d_H,d_H'}$ and $\interval{d_T,d_T'}$, by
maintaining four copies of the state and measuring all the combinations of
$d_H, d_H'$ and $d_T, d_T'$, and subtracting appropriately. This means we are
now maintaining four \emph{different} states for each sample we want, because
of the difference in which states the measurements delete: \[
\frac{1}{\sqrt{M'}}\paren*{\Sc + \sum_{v \in V} \sum_{i = 1}^{\min(d_v, x-1)}
\ket{vih} + \sum_{i = 1}^{\min(d_v, y-1)}\ket{vit}}
\]
where $x = d_H$ or $d_H'$ and $y = d_T$ or $d_T'$.

Unfortunately, the information only about the degrees is not enough to estimate
the bias. Moreover, our method does not automatically permit associating tuples
of integers with each vertex. Our approach is to store the out-degree of the
vertex in the higher order bits of the degree counter.

To do this without overwriting the degree information\footnote{One might ask
why we need to use $d_H'$ and $d_H'$ here, as $d_H$ and $d_T$ would suffice to
avoid a conflict with the degree information. This will become necessary
because of how our measurements delete states from the superposition.}, we will
encode seeing an edge $\dedge{uv}$ by sending $\ket{uih}\to\ket{u(i+d_H')h}$
and $\ket{uit}\to\ket{u(i+d_T')t}$, and then swapping out $d_H' + d_T'$ scratch
states for $\ket{u1h}\dots\ket{ud_H'h}$ and $\ket{u1t}\dots\ket{ud_T't}$. In
order to avoid blowing up the number of states we need, we only do this with
probability $1/d_H$, $1/d_T$ respectively (we will always choose $d_H$ and
$d_H'$ to be within a constant factor of each other), and so if, for instance,
we see $d_u - 1$ edges incident to a vertex $u$ that do not trigger this event,
and then see one that does trigger it, the $\ket{uih}$  portion of our state
becomes \[
\sum_{i = 1}^{\min(d_u, d_H-1)} \ket{uih} + \sum_{u=d_H}^{d_u+d_H'} \ket{uih}
\]
up to normalization. Now, if we measure with \[
\ket{u(d_H+d_H')h} \pm \ket{v(d_T+d_T')t}
\]
in addition to \[
\ket{u(d_H)h} \pm \ket{v(d_T)t}
\]
we can get an estimator that counts when both $u$ and $v$ have degree at least
$d_H$ and $d_T$ respectively, \emph{and} have each seen at least once out-edge
that was sampled. So this approximately checks the out-degree of $u$ and $v$,
and as before we can convert this into approximately counting how often these
out-degrees lie in certain \emph{ranges} by adding three additional estimators.

This method only checks a very rough approximation to the out-degrees and
therefore biases of $u$ and $v$. We improve it by adding extra estimators \[
\ket{u(d_H+id_H')h} \pm \ket{v(d_T+jd_T')t}
\]
for $i = 1,\dots,\kappa-1$ and $j = 1,\dots,\kappa-1$. This has $\kappa^2$ overhead
as we need the measurements to be orthogonal, and so need to perform them on
different pairs of states, but we only need constant $\kappa$ for a
sufficiently accurate estimate. In combination with some ``cleanup''
measurement operators we end up with e.g. the $\ket{uih}$ portion of the state
being \[
\sum_I \sum_{i \in I} \ket{uih} 
\]
up to normalization, where the $I$ are a sequence of intervals, with the last
element in each interval encoding the degree of $u$ and the \emph{number} of
intervals encoding the out-degree of the graph.

\paragraph{Dealing with noise in the bias} The final issue is that our estimate
of the bias of a vertex is noisy, because of the way we count out-edges. One
challenge for snapshot algorithms is that even a small error in estimating 
the bias of a vertex can lead to large errors in the snapshot, if e.g. a vertex
with large degree has bias close to the boundary of a class. The authors
of~\cite{SSSV23b} address this problem by ``smoothing'' techniques. We, on the
other hand, introduce an object that we call the ``pseudosnapshot'',
corresponding to our noisy estimates of the biases, and show that approximating
the entries of this object suffices for the \mdcut\ approximation.

\subsection{Proof overview}
In section \ref{section:dicut_to_bias} we recall the statement of the reduction
from \mdcut{} to computing snapshots. In section
\ref{section:bias_to_pseudobias} we introduce the notion of the pseudosnapshot
and how it is related to the snapshot. In section
\ref{section:pseudobias_quantum_algorithm} we describe the key quantum primitive
that estimates each entry of the pseudosnapshot (restricted to edges with
head and tail in specific ranges of degrees) and prove its correctness.
Finally, in section \ref{section:max_dicut_quantum_algorithm} we show how to
put everything together to get the algorithm for $0.4844$-approximation of
\mdcut.

%% file: preliminaries.tex
\section{Preliminaries}
The graphs we deal with will all be directed graphs. When the graph $G = (V,E)$
being dealt with is clear, $n$ will be the number of the vertices of the graph,
and $m$ the number of edges. For all $v \in V$, $d_v$ is its degree, $\dout_v$
the number of edges with $v$ as their head, $\din_v$ the number with $v$ as
their tail, and $b_v = \frac{\dout_v - \din_v}{d_v}$ will be the \emph{bias} of
$v$.

We will be interested in the \emph{maximum directed cut value} of a graph.
\begin{definition}
Let $G = (V,E)$ be a directed graph. Then \[
\mdcut(G) = \max_{x \in \bool^V}\abs{\set{\dedge{uv} \in E : x_u = 0, x_v = 1}}\text{.}
\]
\end{definition}
Specifically, we will be interested in the complexity of attaining an $\alpha$-approximation to \mdcut.
\begin{definition}
For any $X, X' \in \Rbb_{\ge 0}$, $X'$ $\alpha$-approximates $X$ if $X' \in \brac{\alpha X, X}$.
\end{definition}
Note that recent work on streaming \mcut~\cite{KK19,KP22} uses the opposite
definition, where $X' \in \brac{X, \alpha X}$, so the complexity of achieving
an $\alpha$-approximation for us is equivalent to their achieving a $1/\alpha$
approximation. We adopt this notation for consistency with recent work on
streaming \mdcut{}~\cite{SSSV23b}.

%% file: dicut_to_bias.tex
\section{Reducing \mdcut{} to Snapshot Estimation}\label{section:dicut_to_bias}
We follow~\cite{SSSV23b} in reducing this problem to the problem of estimating
a ``snapshot'' of the graph.
\begin{definition}
Let $\tb \in \brac{-1,1}^\ell$ be a vector of bias thresholds. The (first-order)
snapshot $\hist{G} \in \Nbb^{\ell \times \ell}$ of $G = (V,E)$ is given by: \[
\hist{G}_{i,j} = \abs{\set{\dedge{uv} \in E}: u \in H_i, v \in H_j}
\]
where $H_i$ is the $i\nth$ ``bias class'', given by \[
H_i = \begin{cases}
\set{v \in V : b_v \in \interval{\tb_i,\tb_{i+1}}} & \mbox{$i \in \brac{\ell - 1}$}\\
\set{v \in V : b_v \in \brac{\tb_{\ell}, 1}} & \mbox{$i = \ell$.}
\end{cases}
\]
\end{definition}
\begin{lemma}[From Table~1 of~\cite{S23}, strengthening of Lemma~3.19 of \cite{SSSV23b}, in turn from~\cite{FJ15}]
\label{lm:dicuttobias}
There exists a constant $\alpha > 0.4844$, $\ell \in \Nbb$, a vector of bias
thresholds $\tb \in \brac{-1,1}^{\ell}$, and a vector of probabilities $\rb \in
\brac{0,1}^\ell$ such that $\rb^\dagger \hist{G} (1^{\ell} - \rb)$ is an
$\alpha$-approximation to the \mdcut{} value of $G$.
\end{lemma}

%% file: bias_to_pseudobias.tex
\section{Reducing \mdcut{} to Pseudosnapshot Estimation}\label{section:bias_to_pseudobias}
In this section, we define a ``pseudosnapshot'' $\histps{G}$, based on
hash functions and a coarsening of the vertex degrees, and show that it is
close to the snapshot of a graph $G'$ that is in turn close to $G$. In
the next section, we will show that this can be approximated by a
small space quantum streaming algorithm.
\subsection{Definition}
Let $(d_i)_{i=0}^{\floor{\log_{1+\eps^3}n}}$ be given by $d_i = \floor{(1 +
\eps^3)^i}$ for $i < \floor{\log_{1+\eps^3}  n}$ and $d_{\floor{\log_{1 +
\eps^3}  n}} = n$.  Let $\accuracy \le \poly{n}$, $\eps \in \brac{0,1}$ be
accuracy parameters to be chosen later, and let
$(f_i)_{i=0}^{\floor{\log_{1+\eps^3}  n}}$ be a family of fully independent
random hash functions such that $f_i : E \rightarrow \bool$ is $1$ with
probability $\accuracy/2d_i$, while $g : V \rightarrow \brac{-\eps,\eps}$ is a
fully independent random hash function that is uniform on
$\brac{-\eps,\eps}$.\footnote{We adopt this notation for the sake of clarity,
but we will only ever evaluate $f_i(e)$ at the update when edge $e$ arrives,
and $g$ after processing the entire stream on the endpoints of edges our
algorithm stores, so we do not need to pay the prohibitive overhead of storing
these hash functions. Moreover, while we write $g(e)$ as a random real number,
it will only ever be used in sums and comparisons with numbers of
$\poly(n,\eps)$ precision, so we do not need to store it any more precision
than that.}

Fix an arrival order for the edges $e$ of the directed graph. For any vertex
$v$ and edge $e$, let $\doutb{e}_v$, $\db{e}_v$, refer to the out-degree and
degree of $v$ when only $e$ and edges that arrive before $e$ are counted, and
let $\douta{e}_v$, $\da{e}_v$ refer to these quantities when counting only
edges that arrive \emph{after} $e$. Let $\wt{i}$ be the largest $i$ such that
$d_i < \db{e}$. Then define $\dbps{e}_v = d_{\wt{i}}$, and let $\doutbps{e}_v$
be the number of edges $e'$ with head $v$ that arrive before $e$ and have
$f_{\wt{i}}(e') = 1$, multiplied by $2d_{\wt{i}}/\accuracy$.  We will then define the
$e$-pseudobias of $v$, $\bps{e}_v$, as \[ \min\set*{2\frac{\doutbps{e}_v +
\douta{e}_v}{\dbps{e}_v + \da{e}_v} - 1 + g(v),
1}\text{.}
\]
In other words, the $\bps{e}_v$ is the bias of $v$ when its degree among $e$ and
edges that arrive before $e$ is rounded to the bottom of the interval
$\interval{d_{\wt{i}},d_{\wt{i}+1}}$, and its out-degree among these edges is
estimated using the number of out-edges ``sampled'' by $f_{\wt{i}}$, with a
small amount of noise $g(v)$ added. Since this can sometimes produce a
pseudobias larger than 1, we then cap it at 1.

\begin{definition}
Let $\tb \in \brac{-1,1}^\ell$ be a vector of bias thresholds. The
pseudosnapshot $\histps{G} \in \Nbb^{\ell \times \ell}$ of $G = (V,E)$ is given
by: \[
\histps{G}_{i,j} = \abs{\set{\dedge{uv} \in E: u \in H^{\dedge{uv}}_i, v \in
H^{\dedge{uv}}_j}}
\]
where $H^e_i$ is the $i\nth$ ``$e$-pseudobias'' class, given by \[
H^e_i = \begin{cases}
\set{v \in V : \bps{e}_v \in \interval{\tb_i,\tb_{i+1}}} & \mbox{$i \in \brac{\ell - 1}$}\\
\set{v \in V : \bps{e}_v \in \brac{\tb_{\ell}, 1}} & \mbox{$i = \ell$.}
\end{cases}
\]
The restriction of $\histps{G}$ to $E' \subseteq E$ is then given by: \[
\histps{G,E'}_{i,j} =  \abs{\set{\dedge{uv} \in E': u \in H^{\dedge{uv}}_i, v \in
H^{\dedge{uv}}_j}}
\]
\end{definition}

\subsection{Closeness to Snapshot}
In this section we show that the snapshot and pseudosnapshot of a graph are
close enough to each other for the purpose of approximating \mdcut. We will
assume throughout that the snapshot and pseudosnapshot are defined relative to
the same vector of bias thresholds $\tb$. We will refer to the intervals
$(\interval{\tb_i, \tb_i})_{i \in \brac{\ell - 1}}$ and $\interval{\tb_\ell,
1}$ as ``bias intervals''. 

We start by showing that, with good enough probability, most of the edges
incident to any vertex $v$ will give $v$ a pseudobias in the same interval, and
that pseudobias will not be too far from the true bias.
\begin{lemma}
\label{lm:bpsvacc}
For each $v \in V$, with probability $1 - \bO{\eps^2} - e^{-\bO{\eps^6
\accuracy}}$ over the hash functions $(f_i)_{i = 0}^{\floor{\log_{1+\eps^3}n}}$
and $g$, $\bps{e}_v$ is in the same bias interval for all but $\eps d_v$ of the
edges $e$ incident to $v$, and all of these $\bps{e}_v$ are within $\bO{\eps}$
of $b_v$.  Moreover, these events depend only on $g(v)$ and $f(e)$ for edges
with head $v$.
\end{lemma}
\begin{proof}
As there are only constantly many bias intervals, we may without loss of
generality assume that $\eps$ is smaller than half the distance between the
boundaries of any bias interval. Then for every $v \in V$, with probability $1 -
\bO{\eps^2}$ over $g(v)$, $b_v + g(e)$ is at least $C\eps^3$ away from the boundary
of any bias interval, for $C > 0$ a constant to be chosen later. It will
therefore suffice for $\abs{\bps{e}_v - b_v - g(e)} < C\eps^3$ to hold for all
but $\eps d_v$ of the edges $e$ incident to $v$.

First note that, by the definition of the $d_i$, $\dbps{e}_v \in
\interval*{\frac{\db{e}_v}{1 + \eps^3}, \db{e}_v}$ for all $e$. So we only need
to bound $\doutbps{e}_v$.

We will ignore the first $\eps d_v$ edges to arrive incident to $v$. For the edges $e$
arriving after this, $\db{e}_v$ passes through only $\bO{\log_{1 + \eps^3}
\frac{1}{\eps}} = \bO{\frac{1}{\eps^3}\log\frac{1}{\eps}}$ intervals
$\interval{d_i, d_{i+1}}$. It will suffice to show that, for each such interval,\[
\doutbps{e}_v \in ((1 - \eps^3)\doutb{e}_v, (1 + \eps^3)\doutb{e}_v)
\]
for the $d_{i}\nth$ $e$ to arrive incident to $v$ (as $\doutbps{e}_v$ only
increases, and $\doutb{e}_v$ only changes by a $(1 + \eps^3)$ multiplicative factor
in this interval).

So fix such an $e$. Then $\doutbps{e}_v$ is $2d_i/\accuracy$ times the number of
edges $e'$ that arrive before $e$ (including $e$) and have $f_i(e') = 1$. So it is distributed as \[
\frac{2d_i}{\accuracy}\text{B}\paren*{\doutb{e}_v, \frac{\accuracy}{2d_i}}
\]
and so by the Chernoff bounds it is within $\eps^3 \doutb{e}_v$ of
$\doutb{e}_v$ with probability \[ 
1 - e^{-\bO{\eps^6 \accuracy \frac{\doutb{e}_v}{d_i}}} =  - e^{-\bO{\eps^6
\accuracy}}
\]
and by taking a union bound over the $\bO{\frac{1}{\eps^3}\log\frac{1}{\eps}}$
intervals, we have that, with probability $1 - e^{-\bO{\eps^6 \accuracy}}$ over
$(f_i)_{i = 0}^{\floor{\log_{1+\eps^3}n}}$, $\dbps{e}_v$ and $\doutbps{e}_v$ are
$(1 + \bO{\eps^3})$ multiplicative approximations of $\db{e}_v$ and
$\doutb{e}_v$, respectively, and so \[
\bps{e}_v \in (b_v + g(v) - \bO{\eps^3}, b_v + g(v) + \bO{\eps^3})
\] 
for all $e$ incident to $v$ after the first $\eps d_v$ to arrive. The lemma
therefore holds if we choose $C$ to be a large enough constant.
\end{proof}
This allows us to show that, with only small edits to $G$, we can change the
bias of vertices in $G$ in such a way that the pseudosnapshot becomes an
approximately accurate snapshot.
\begin{lemma}
\label{lm:pseudobiascloseness}
With probability $1 - \bO{\eps} - e^{-\bO{\eps^6\accuracy}}$ over $(f_i)_{i =
0}^{\floor{\log_{1+\eps^3} n}}$ and $g$, there exists a graph $G'$ differing
from $G$ in $\bO{\eps m}$ edges such that any $\bO{\eps m}$-accurate estimate
of the pseudosnapshot of $G$ is a $\bO{\eps m}$-accurate estimate of the
snapshot of $G'$.
\end{lemma}
\begin{proof}
Let $V'$ be the set of all vertices $V$ for which the event described in
Lemma~\ref{lm:bpsvacc} occurs. The vertex set of $G'$ will be $V' \cup \set{\perp}$,
where $\perp$ is a newly introduced dummy vertex. We will construct $G'$ by first
removing all vertices in $V \setminus V'$ from $G$, and replacing the edges
between them and $V'$ with corresponding edges between $\perp$ and $V'$.

Then, for each vertex $v \in V'$, we will add $\bO{\eps d_v}$ edges between $v$
and $\perp$, choosing the orientation of these edges so that the bias of $v$ in
$G'$ is in the same bias interval as $\bps{e}_v$ for all but $\eps d_v$ of the
edges $e$ incident to $v$.

$G'$ differs from $G$ in at most \[
2\sum_{v \in V \setminus V'}d_v + \sum_{v \in V'} \bO{\eps d_v}
\]
edges. So by Markov's inequality, as each edge is in $V \setminus V'$ with
probability at most $\bO{\eps^2} + e^{-\bO{\eps^6\accuracy}}$, this is
$\bO{\eps m}$ with probability at least $1 - \bO{\eps} -
e^{-\bO{\eps^6\accuracy}}$.

Moreover, by Lemma~\ref{lm:bpsvacc} all but $\sum_{v \in V \setminus V'}d_v +
\sum_{v \in V}\eps d_v$ of the edges counted in the pseudosnapshot of $G$
had endpoints whose pseudobiases were in the same bias intervals as the biases of
the corresponding vertices in $G'$. Furthermore, there are only $\sum_{v \in V
\setminus V'}d_v + \sum_{v \in V'} \bO{\eps d_v}$ edges in $G'$ that were
\emph{not} counted in the pseudosnapshot of $G$---those added to replace
edges between $V$ and $V'$, and those added to correct the biases of edges in
$V'$. So again by Markov, the pseudosnapshot of $G$ is a $\bO{\eps
m}$-accurate snapshot for $G'$ with probability at least $1 - \bO{\eps} -
e^{-\bO{\eps^6\accuracy}}$.
\end{proof}

\subsection{Reduction from \mdcut}
Now, we can show that the pseudosnapshot gives a good \mdcut{} estimate, by
using the fact that the $G'$ constructed in the previous section has a similar
\mdcut{} value to $G$.
\begin{lemma}
\label{lm:dicuttopseudobias}
Let $\alpha$, $\ell$, $\tb$, $\rb$ be as in Lemma~\ref{lm:dicuttobias}. Then there
exists a constant $C > 0$ such that \[
\rb^\dagger \histps{G} (1^\ell - \rb) - C \eps m
\]
is an $(\alpha - \bO{\eps})$-approximation to $\mdcut\paren{G}$ with probability $1
- \bO{\eps} - e^{-\bO{\eps^6\accuracy}}$ over $(f_i)_{i =
0}^{\floor{\log_{1+\eps^3} n}}$ and $g$.
\end{lemma}
\begin{proof}
Suppose the probability $1 - \bO{\eps} - e^{-\bO{\eps^6\accuracy}}$ event of
Lemma~\ref{lm:pseudobiascloseness} holds. Then there exists a graph $G'$ that
differs from $G$ in only $\eps m$ edges such that \[
\abs{\histps{G} - \hist{G'}} = \bO{\eps m}
\]
and so \[
\abs{\rb^\dagger \histps{G} (1^\ell - \rb) - \rb^\dagger \hist{G'} (1^\ell - \rb)} = \bO{\eps m}
\]
which by Lemma~\ref{lm:dicuttobias} implies that \[
\alpha \mdcut\paren{G'} - \bO{\eps m} \le \rb^\dagger \histps{G} (1^\ell - \rb) \le \mdcut\paren{G'} + \bO{\eps m}.
\]
As $G'$ only differs from $G$ in $\bO{\eps m}$ edges, \[
\abs{\mdcut\paren{G'} - \mdcut\paren{G}} = \bO{\eps m}
\]
and so the result follows by setting $C$ to be a large enough constant that \[
\rb^\dagger \histps{G} (1^\ell - \rb) - C\eps m \le \mdcut\paren{G}.
\]
\end{proof}

%% file: pseudobias_quantum_algorithm.tex
\section{Quantum Algorithm for Pseudosnapshot Estimation}\label{section:pseudobias_quantum_algorithm}
In this section we describe our key quantum primitive: a small-space algorithm
for estimating the pseudosnapshot of a directed graph $G = (V,E)$, restricted
to edges $e = \dedge{uv}$ with $\db{e}_u \in \interval{d_i,d_{i+1}}$, $\db{e}_v
\in \interval{d_j, d_{j+1}}$ for some specified $i, j \in
\brac{\floor{\log_{1+\eps^3} n}} \cup \set{0}$.
\subsection{Algorithm}
\begin{restatable}{lemma}{maxcutlb}
\label{lm:dicut-sketch}
Let $\alpha, \beta \in \brac{\floor{\log_{1+\eps^3}  n}} \cup \set{0}$. Let
$\accuracy = \poly(n)$ be the integer accuracy parameter used in defining the
pseudobiases. Fix a draw of the hash functions $(f_i)_{i =
0}^{\ceil{\log_{1+\eps^3} n}}$, and therefore the pseudobiases of the graph
$G$.

Then there is a quantum streaming algorithm that, if $\sum_{e \in E}
(f_\alpha(e) + f_\beta(e)) \le 2\accuracy m $, returns an estimate of the
pseudosnapshot of $G$ restricted to edges $e = \dedge{uv}$ with $\db{e}_u \in
\interval{d_\alpha,d_{\alpha+1}}$, $\db{e}_v \in \interval{d_\beta,
d_{\beta+1}}$.

Each entry of the estimate has bias at most the number of edges
$\dedge{uv}$ such that:
\begin{enumerate}
\item $\db{\dedge{uv}}_u \in \interval{d_\alpha,
d_{\alpha+1}}$
\item $\db{\dedge{uv}}_v \in \interval{d_\beta, d_{\beta+1}}$
\item $\max\set*{\frac{\accuracy}{2d_{\alpha}} \doutbps{\dedge{uv}}_u + 1,
\frac{\accuracy}{2d_{\beta}} \doutbps{\dedge{uv}}_v + 1} > \accuracy$
\end{enumerate}

Each entry has variance $\bO{\accuracy^3m^2}$. The algorithm uses $\bO{\log n}$
qubits of space. 
\end{restatable}

The algorithm will maintain a superposition \[
\frac{\ket{\Sc} + \sum_{i=1}^{2\accuracy^2}\paren*{\ket{\Ac^i} + \ket{\Bc^i} +
\ket{\Cc^i} + \ket{\Dc^i}}}{\sqrt{M'}}
\]
where $M' \le M = C\accuracy^3m$ for some sufficiently large constant $C$. $M'$ will
start at $M$ and decrease as measurements remove states from the superposition.

A key primitive for maintaining this state will be ``swapping'', in which we
execute the unitary that swaps two named basis states while leaving the rest of
the Hilbert space unchanged. 

$\ket{\Sc}$ will be our ``scratch states''. It is equal to \[
\sum_{i=t}^M \ket{is}
\]
where $t$ starts at $1$ and is incremented every time we use a scratch state
(by swapping it with some other state we want), and the last register $s$ 
indicates that this is a scratch state. The algorithm will keep track of $t$ 
so that it knows which state to swap from.

The remaining vectors encode information about vertices and their degrees. For
$\Ec = \Ac, \Bc, \Cc, \Dc$, each takes the form
\begin{align*}
\ket{\Ec^i} &= \sum_{v \in V}\ket{\Ec_v^i}\\
\ket{\Ec_v^i} &= \sum_{j \in E_v}\ket{vje}
\end{align*}
where the sets $E_v = A_v, B_v, C_v, D_v$ are not explicitly stored by the
algorithm but are updated by updating the superposition, and $e = a^i, b^i,
c^i, d^i$ marks which of $\Ac^i, \Bc^i, \Cc^i, \Dc^i$ a state belongs to, for
$i \in \brac{2\accuracy^2}$. Note that the underlying set $E_v$ for each
$\Ec_v^i$ does not depend on $i$---the $2\accuracy^2$ copies are to allow us to
perform a larger set of measurements, and the states will remain identical
until the algorithm terminates.  The four non-scratch components are broken
into two pairs (as each vertex can be both a head and a tail, and the two cases
need to be handled separately), with each pair having an element for checking
whether degrees are high enough, and one for checking whether degrees are low
enough (both will track whether vertices have had enough edges coming out of
them that pass the hash functions).
\begin{itemize}
\item $\Ac^i$ are for tracking vertices with degree at least $d_\alpha$.
\item $\Bc^i$ are for tracking when those vertices have their degree exceed
$d_{\alpha+1}$.
\item $\Cc^i$ are for tracking vertices with degree at least $d_\beta$.
\item $\Dc^i$ are for tracking when those vertices have their degree exceed
$d_{\beta+1}$.
\end{itemize}
At the start of the execution of our algorithm, $A_v, B_v, C_v, D_v =
\emptyset$. They will be updated with the following three operations:
\begin{itemize}
\item $\inc(\Ec, v, r)$ replaces $E_v$ with $\set{i + r : i \in E_v} \cup
\brac{r}$, and replaces $t$ with $t + 2\accuracy^2r$. Note that this can be accomplished
with a single unitary transformation, by sending $\ket{vie^j}$ to $\ket{v(i+r)e^j}$
for each $i,j$, and then swapping the first $2\accuracy^2r$ remaining elements of $\Sc$ with
$\sum_{i=1}^r \ket{vie^j}$ for each $j \in \brac{2\accuracy^2}$.
\item $\measure(u,v)$ measures the superposition with projectors onto the
following vectors for each $(i,j) \in \brac{\accuracy}^2$, and $b \in \bool$, along
with a projection onto the complement of the space spanned by the projectors.
\begin{align*}
\ket{\nu_{i,j,b}^1} &= \frac{\ket{u(d_\alpha + (i-1)d_{\alpha+1})a^{r_{i,j}}} + (-1)^b
\ket{v(d_\beta + (j-1)d_{\beta+1})c^{r_{i,j}}}}{\sqrt{2}}\\
\ket{\nu_{i,j,b}^2} &= \frac{\ket{u(id_{\alpha+1})b^{r_{i,j}}} + (-1)^b \ket{v(d_\beta
+ (j-1)d_{\beta+1})c^{s_{i,j}}}}{\sqrt{2}}\\
\ket{\nu_{i,j,b}^3} &= \frac{\ket{u(d_\alpha + (i-1)d_{\alpha+1})a^{s_{i,j}}} + (-1)^b
\ket{v(jd_{\beta+1})d^{r_{i,j}}}}{\sqrt{2}}\\
\ket{\nu_{i,j,b}^4} &= \frac{\ket{u(id_{\alpha+1})b^{s_{i,j}}} + (-1)^b
\ket{v(jd_{\beta+1})d^{s_{i,j}}}}{\sqrt{2}}
\end{align*} 
Where the $s_{i,j}$, $r_{i,j}$ are chosen so that these vectors are all
orthogonal to one another across all $i,j,b$ (note that this is possible
because we have $\accuracy^2$ possible values to choose from). If the result of
this measurement is anything other than the complementary projection, the
quantum part of the algorithm will terminate and the remaining execution will
be entirely classical.
\item $\cleanup(u,v)$ measures the superposition with projectors onto
$\ket{w(d_\alpha + id_{\alpha+1})a^{j}}$, $\ket{w((i+1)d_{\alpha+1})b^{j}}$,
$\ket{w(d_\beta + id_{\beta+1})c^{j}}$, and $\ket{w((i+1)d_{\beta+1})d^{j}}$,
for all $i \in \brac{M}$, $j \in \brac{2\accuracy^2}$, and $w = u,v$, along with a
projector onto the complement of the space they span. If anything other than
the complementary projector is returned, the algorithm halts entirely and
outputs a zero estimate for the pseudosnapshot.

Together, the effect of performing $\measure(u,v)$ and $\cleanup(u,v)$, if they
do \emph{not} return something other than the complementary projectors, is to
delete the following elements from $A_w$, $B_w$, $C_w$, $D_w$, for $w = u,v$
and for all $i \in \brac{M}$ (note that these elements may have not been
present to begin with, or may be ``removed'' multiple times between the two
operations---this does not cause any issues):
\begin{itemize}
\item $d_\alpha + (i-1)d_{\alpha + 1}$ from $A_w$.
\item $id_{\alpha + 1}$ from $B_w$.
\item $d_\beta + (i-1)d_{\alpha + 1}$ from $C_w$.
\item $id_{\beta + 1}$ from $D_w$.
\end{itemize}
\end{itemize}
We can now describe the algorithm.

\paragraph{Quantum Stage} For each $\dedge{uv}$ processed until the quantum
stage terminates:
\begin{enumerate}
\item $\inc(\Ec,w,1)$ for $w = u,v$, and $\Ec = \Ac, \Bc, \Cc, \Dc$.
\item If $f_\alpha(\dedge{uv}) = 1$, $\inc(\Ac,u,d_{\alpha+1})$ and
$\inc(\Bc,u,d_{\alpha+1})$.
\item If $f_\beta(\dedge{uv}) = 1$, $\inc(\Cc,u,d_{\beta + 1})$ and
$\inc(\Dc,u,d_{\beta+1})$.
\item $\measure(u,v)$. If the measurement returns $\ket{\nu_{i,j,b}^r}$, pass it
along with $u,v$ to the classical stage and continue.
\item $\cleanup(u,v)$. If the measurement returns anything other than the
projector onto the complement of the cleanup vectors, immediately terminate the
algorithm, outputting an all-zeroes estimate.
\end{enumerate}
If the quantum stage processes every edge without being terminated by a
measurement outcome, output an all-zeroes estimate and skip the classical
stage.
\paragraph{Classical Stage} 
For the remainder of the stream, track $\douta{e}_u$,$\douta{e}_v$,
$\da{e}_u$,$\da{e}_v$ (giving us exact values for these variables). Then
estimate $\db{e}_u$, $\db{e}_v$ by assuming that they are equal to $d_\alpha$,
$d_\beta$, respectively. Then estimate $\doutbps{e}_u$, $\doutbps{e}_v$, by
assuming that the number of edges $e$ with head $u$ and $f_\alpha(e) = 1$ is
$i-1$, and the number with head $v$ and $f_\beta(e) = 1$ is $j-1$.

Combine these estimates and evaluate $g(u)$, $g(v)$ to estimate $\bps{e}_u$ and
$\bps{e}_v$. If $r = 1$ or $4$, set $(-1)^bM/2$ as the corresponding entry of
the pseudosnapshot estimate (with every other entry as $0$). If $r = 2$ or $3$,
set it as $(-1)^{\overline{b}}M/2$ instead. 
\subsection{Analysis}
\subsubsection{State Invariant}
\begin{lemma}
\label{lm:stinv}
Consider any time after some number of edges have been (completely) processed
in the quantum stage, and suppose the state has not yet terminated, and $t$ has
not exceeded $M$. Let $v \in V$, and let $r$ be the number of those edges that
were incident to $v$. Let $R$ be the number of those edges $e$ such that $v$
was the head of the edge, and $f_\alpha(e) = 1$. 

Then, if $R = 0$, 
\begin{align*}
A_v &= \Nbb \cap \interval{1,\min(r+1,d_\alpha)}\\
B_v &= \Nbb \cap \interval{1,\min(r+1,d_{\alpha +1})})
\end{align*}
and if $R > 0$, there exists $(\prefix_i)_{i=1}^R \in \brac{r}^R$ such that:
\begin{align*}
A_v &= \Nbb \cap \paren*{\interval{1,d_\alpha} \cup \bigcup_{i=1}^R I_i}\\
B_v &= \Nbb \cap \paren*{\interval{1,d_{\alpha+1}} \cup \bigcup_{i=1}^R J_i}
\end{align*}
Where
\begin{align*}
I_i &= \begin{cases}
\interval{d_\alpha+(i-1)d_{\alpha+1}+\prefix_i,d_\alpha + id_{\alpha + 1}} & \mbox{i
< R}\\
\interval{d_\alpha + (R-1)d_{\alpha + 1} + \prefix_R, Rd_{\alpha + 1}
+ \min(r+1, d_{\alpha})} & \mbox{i = R}
\end{cases}\\
J_i &= \begin{cases}
\interval{id_{\alpha+1}+\prefix_i,(i+1)d_{\alpha + 1}} & \mbox{i
< R}\\
\interval{Rd_{\alpha + 1} + \prefix_R, Rd_{\alpha + 1}
+ \min(r+1, d_{\alpha+1})} & \mbox{i = R}
\end{cases}
\end{align*}
The same relationship holds for $C_v$ and $D_v$, except with $\beta$ instead of $\alpha$.
\end{lemma}
\begin{proof} We will prove the result for $A_v$ and $B_v$. The proof for $C_v$ 
and $D_v$ is identical, with $\beta$ substituted for $\alpha$. Note that when
describing the updates performed on seeing an edge $\dedge{uv}$ or $\dedge{vu}$
we will ignore the updates that only touch $u$, as they have no effect on the
sets we are analyzing here.

We prove this result by induction. First, we consider $R = 0$ and $r = 0$ as
the base case. Next, we prove the inductive step where $r$ is increased by 1
while $R = 0$. This finishes the proof for $R = 0$ and any $r$. Next, we prove
the result for $R = 1$ by considering the step when $R$ is increased from $R =
0$. Note that any update that increases $R$ also increases $r$, so in this case
$r$ is also increased by 1. Last, we fix any $r$ and any $R > 0$ and prove the
inductive step when $r$ is increased by 1, but $R$ stays unchanged, and when
both $r$ and $R$ are increased by 1, which completes the proof.

We start with the case where $R = 0$ and $r =
0$. Then the statement is equivalent to $A_v = B_v = \emptyset$, which follows from
how we have defined the initial state of the algorithm, and the fact that edges
not incident to $v$ do not result in updates that affect $A_v$ or $B_v$. 

Now suppose the result holds for $r$ (with $R$ still $0$), and consider the $(r
+ 1)\nth$ edge incident to $v$ to arrive (with no $\dedge{vw}$ such that
$f(\dedge{vw}) = 1$ having arrived yet). Then the update consists of performing
$\inc(\Ac,v,1)$, $\inc(\Bc,v,1)$, and then $\measure(v, w)$ or $\measure(w,v)$
for some $w$, followed by $\cleanup(v,w)$ or $\cleanup(w,v)$. Before the
measurement, this gives us
\begin{align*}
A_v &= \Nbb \cap \interval{1,\min(r+2,d_\alpha + 1)}\\
B_v &= \Nbb \cap \interval{1,\min(r+2,d_{\alpha + 1} + 1)})
\end{align*}
and as the condition of the lemma supposes that the measurement did not result
in the algorithm terminating, it will have removed $d_\alpha$ from $A_v$ and
$d_{\alpha+1}$ from $B_v$, and so 
\begin{align*}
A_v &= \Nbb \cap \interval{1,\min(r+2,d_\alpha)}\\
B_v &= \Nbb \cap \interval{1,\min(r+2,d_{\alpha + 1})})
\end{align*}
as desired.

Next we will prove the result for $R$ increased from $0$ to $1$ and $r$ increased by $1$. By the previous section, when $R = 0$ the state is
\begin{align*}
A_v &= \Nbb \cap \interval{1,\min(r+1,d_\alpha)}\\
B_v &= \Nbb \cap \interval{1,\min(r+1,d_{\alpha + 1})})\text{.}
\end{align*}
When the $(r+1)\nth$ edge incident to $v$ arrives, if it is also the first edge
with head $v$ in $f_{\alpha}^{-1}(1)$, the update consists of performing
$\inc(\Ac,v,1)$, $\inc(\Bc,v,1)$, $\inc(\Ac,v,d_{\alpha + 1})$,
$\inc(\Bc,v,d_{\alpha + 1})$, and then $\measure(v, w)$ for
some $w$, followed by $\cleanup(v,w)$. 

After the $\inc$ operations are performed we have 
\begin{align*}
A_v &= \Nbb \cap \interval{1,\min(r+d_{\alpha + 1} + 2,d_\alpha + d_{\alpha +
1} + 1)}\\
B_v &= \Nbb \cap \interval{1,\min(r+d_{\alpha+1} + 2,2d_{\alpha + 1} +
1)})
\end{align*}
and then after performing the measurements we remove $d_\alpha$ and $d_\alpha +
d_{\alpha+1}$ from $A_v$ and $d_{\alpha + 1}$ and $2d_{\alpha + 1}$ from $B_v$,
so we have
\begin{align*}
A_v &= \Nbb \cap (\interval{1,d_\alpha} \cup \interval{d_\alpha + 1,
\min(r+d_{\alpha + 1} + 2,d_\alpha + d_{\alpha + 1})})\\
B_v &= \Nbb \cap (\interval{1,d_{\alpha+1}} \cup \interval{d_{\alpha+1} + 1,
\min(r+d_{\alpha+1} + 2,2d_{\alpha + 1})}))
\end{align*}
which matches the lemma statement by setting $\prefix_1 = \prefix_R = 1$.

We now need to consider two more cases to complete the proof: First, the case when we have the result for
some $R > 0$ and $r$, and the update consists of the $(r+1)\nth$ edge incident
to $v$, and this edge is \emph{not} in $f_{\alpha}^{-1}(1)$ or does not have
$v$ as its head.  Second, the case when we have the result for $R$ and $r$ and
the update consists of the $(r+1)\nth$ edge incident to $v$, which is also the
$(R+1)\nth$ edge with head $v$ in $f_{\alpha}^{-1}(1)$.

Let us deal with the cases when $r$ is increased from $r$ to $r+1$ and $R$ stays the same first. The update consists of
performing $\inc(\Ac,v,1)$, $\inc(\Bc,v,1)$, and then $\measure(v, w)$ or
$\measure(w,v)$ for some $w$, followed by $\cleanup(v,w)$ or
$\cleanup(w,v)$. The state before the update, by the inductive hypothesis, is
\begin{align*}
A_v &= \Nbb \cap \paren*{\interval{1,d_\alpha} \cup \bigcup_{i=1}^R I_i}\\
B_v &= \Nbb \cap \paren*{\interval{1,d_{\alpha+1}} \cup \bigcup_{i=1}^R J_i}
\end{align*}
with the intervals $I_i$ and $J_i$ as defined in the lemma statement. The
$\inc(\Ac, v, 1)$ and $\inc(\Bc, v, 1)$ operations, then, correspond to
replacing these intervals with
\begin{align*}
I_i' &= \begin{cases}
\interval{d_\alpha+(i-1)d_{\alpha+1}+\prefix_i + 1,d_\alpha + id_{\alpha + 1} + 1} & \mbox{i
< R}\\
\interval{d_\alpha + (R-1)d_{\alpha + 1} + \prefix_R + 1, Rd_{\alpha + 1}
+ \min(r+2, d_{\alpha} + 1) } & \mbox{i = R}
\end{cases}\\
J_i' &= \begin{cases}
\interval{id_{\alpha+1}+\prefix_i + 1,(i+1)d_{\alpha + 1} + 1} & \mbox{i
< R}\\
\interval{Rd_{\alpha + 1} + \prefix_R + 1, Rd_{\alpha + 1}
+ \min(r+2, d_{\alpha+1} + 1)} & \mbox{i = R}
\end{cases}
\end{align*}
and $\interval{1,d_\alpha}$ with $\interval{1,d_\alpha + 1}$, and
$\interval{1,d_{\alpha+1}}$ with $\interval{1,d_{\alpha+1}+1}$.

The measurements then delete $d_\alpha + (i-1)d_\alpha$ from $A_v$ and
$id_\alpha$ from $B_v$ for every $i \in [M]$. So $\interval{1,d_\alpha + 1}$ and
$\interval{1,d_{\alpha+1}}$ return to $\interval{1,d_\alpha}$,
$\interval{1,d_{\alpha+1}}$, respectively, and the intervals $I_i$, $J_i$
become
\begin{align*}
I_i'' &= \begin{cases}
\interval{d_\alpha+(i-1)d_{\alpha+1}+\prefix_i + 1,d_\alpha + id_{\alpha + 1}} & \mbox{i
< R}\\
\interval{d_\alpha + (R-1)d_{\alpha + 1} + \prefix_R + 1, Rd_{\alpha + 1}
+ \min(r+2, d_{\alpha}) } & \mbox{i = R}
\end{cases}\\
J_i'' &= \begin{cases}
\interval{id_{\alpha+1}+\prefix_i + 1,(i+1)d_{\alpha}} & \mbox{i
< R}\\
\interval{Rd_{\alpha + 1} + \prefix_R + 1, Rd_{\alpha + 1}
+ \min(r+2, d_{\alpha+1})} & \mbox{i = R}
\end{cases}
\end{align*}
and so by setting the new $I_i$ to be $I_i''$ and likewise with $J_i$, the
states are in the form desired (by incrementing every element of
$(\prefix_i)_{i=1}^R$ by 1), as $r$ is now $1$ larger.

Finally, we consider the inductive step where both $R$ and $r$ are increased by
$1$. This means the arriving edge is the $(r+1)\nth$ edge incident to $v$,
which is also the $(R+1)\nth$ edge with head $v$ in $f^{-1}_\alpha(1)$.

 The update consists of performing $\inc(\Ac,v,1)$, $\inc(\Bc,v,1)$,
 $\inc(\Ac,v,d_{\alpha + 1})$, $\inc(\Bc,v,d_{\alpha + 1})$, and then
 $\measure(v, w)$ for some $w$, followed by $\cleanup(v,w)$. The state before
 the update, by the inductive hypothesis, is
\begin{align*}
A_v &= \Nbb \cap \paren*{\interval{1,d_\alpha} \cup \bigcup_{i=1}^R I_i}\\
B_v &= \Nbb \cap \paren*{\interval{1,d_{\alpha+1}} \cup \bigcup_{i=1}^R J_i}
\end{align*}
with the intervals $I_i$ and $J_i$ as defined in the lemma statement. As $\inc$
operations add together, after all the increment operations we have replaced
$\interval{1,d_\alpha}$ with $\interval{1,d_\alpha + d_{\alpha + 1} + 1}$ and
$\interval{1,d_{\alpha+1}}$ with $\interval{1,2d_{\alpha+1}  + 1}$ and $I_i$,
$J_i$ with 
\begin{align*}
I_i' &= \begin{cases}
\interval{d_\alpha+id_{\alpha+1}+\prefix_i + 1,d_\alpha + (i+1)d_{\alpha + 1} + 1}
& \mbox{i < R}\\
\interval{d_\alpha + Rd_{\alpha + 1} + \prefix_R + 1, (R+1)d_{\alpha + 1}
+ \min(r+2, d_{\alpha} + 1) } & \mbox{i = R}
\end{cases}\\
J_i' &= \begin{cases}
\interval{(i+1)d_{\alpha+1}+\prefix_i + 1,(i+2)d_{\alpha + 1} + 1} & \mbox{i
< R}\\
\interval{(R+1)d_{\alpha + 1} + \prefix_R + 1, (R+1)d_{\alpha + 1}
+ \min(r+2, d_{\alpha+1} + 1)} & \mbox{i = R}
\end{cases}
\end{align*}
and so after the measurements, as they remove $d_{\alpha} + (i-1)d_{\alpha +
1}$ from $A_v$ and $id_{\alpha + 1}$ from $B_v$ for all $i$,
$\interval{1,d_\alpha + d_{\alpha + 1} + 1}$ becomes $\interval{1,d_\alpha}
\cup \interval{d_\alpha + 1, d_\alpha + d_{\alpha + 1}}$ and
$\interval{1,2d_{\alpha+1}  + 1}$ becomes $\interval{1,d_{\alpha+1}} \cup
\interval{d_{\alpha + 1}  + 1, 2d_{\alpha +1}}$. $I_i'$ and $J_i'$ become
\begin{align*}
I_i'' &= \begin{cases}
\interval{d_\alpha+id_{\alpha+1}+\prefix_i + 1,d_\alpha + (i+1)d_{\alpha + 1}}
& \mbox{i < R}\\
\interval{d_\alpha + Rd_{\alpha + 1} + \prefix_R + 1, (R+1)d_{\alpha + 1}
+ \min(r+2, d_{\alpha}) } & \mbox{i = R}
\end{cases}\\
J_i'' &= \begin{cases}
\interval{(i+1)d_{\alpha+1}+\prefix_i + 1,(i+2)d_{\alpha + 1}} & \mbox{i
< R}\\
\interval{(R+1)d_{\alpha + 1} + \prefix_R + 1, (R+1)d_{\alpha + 1}
+ \min(r+2, d_{\alpha+1})} & \mbox{i = R}
\end{cases}
\end{align*}
and so the state is of the form required, by setting (writing $\prefix_i'$ for the
old $\prefix_i$, and noting that this means $\prefix_i \le r+1$ for all $i$)
\begin{align*}
I_i &= \begin{cases}
\interval{d_\alpha + 1, d_\alpha + d_{\alpha + 1}} & \mbox{$i = 1$}\\
I_{i-1}'' &\mbox{$1 < i \le R+1$}
\end{cases}\\
J_i &= \begin{cases}
\interval{d_{\alpha + 1}  + 1, 2d_{\alpha +1}} & \mbox{$i = 1$}\\
J_{i-1}'' &\mbox{$1 < i \le R+1$}
\end{cases}\\
\prefix_i &= \begin{cases}
1 & \mbox{i = 1}\\
\prefix_{i - 1}' + 1 & \mbox{$1 < i \le R+1$}
\end{cases}
\end{align*}
which completes the proof.
 \end{proof}

\subsubsection{Measurement Outcomes}
In this section we characterize the expected effect on the pseudosnapshot
estimate from all of the measurement outcomes that can terminate the quantum
stage of the algorithm: those other than the residual projector in $\measure$
and $\cleanup$.
\begin{lemma}
For all $u,v$, the expected contribution of $\cleanup(u,v)$ to every entry of
the pseudosnapshot estimate is $0$.
\end{lemma}
\begin{proof}
This follows immediately from the fact that we do not modify the estimate when
the algorithm terminates due to $\cleanup$.
\end{proof}
In order to analyze the expectation of the estimate output by the algorithm, we
will need the following lemma about the probability of the algorithm
terminating before a certain point. Note that for any stream of edges, the
sequence of measurements performed by the algorithm is deterministic, except
that the sequence may be terminated early depending on the measurement results.
\begin{lemma}
\label{lm:earlyterm}
For any $k \in \Nbb$, let $p_k$ be the probability that the algorithm
terminates in the first $k$ measurements performed by the algorithm. Then, if
the algorithm does not terminate, the value of $M'$ after these measurements is
$(1-p_k)M$.
\end{lemma}
\begin{proof}
We will prove a slightly stronger version of the result, in which we treat
$\measure$ as $4\accuracy^2$ sequential measurements on \[
\ket{\nu_{i,j,0}^a}\bra{\nu_{i,j,0}^a}, \ket{\nu_{i,j,1}^a}\bra{\nu_{i,j,1}^a},
\mathbbm{1} - \ket{\nu_{i,j,0}^a}\bra{\nu_{i,j,0}^a} -
\ket{\nu_{i,j,1}^a}\bra{\nu_{i,j,1}^a}
\]
for $a \in \brac{4}$, $(i,j) \in \brac{\accuracy}^2$, and $\cleanup$ as $16M\accuracy^2$
sequential measurements on \[
P, \mathbbm{1} - P
\]
for each of the projectors $P$ used in $\cleanup$. The lemma result is then just a
special case in which $k$ is restricted to those values corresponding to a
discrete set of $\measure$ and $\cleanup$ measurements.

We proceed by induction on $k$. For $k=0$, $p_k = 0$ and $M' = M$. So suppose
that after $k$ measurements, $M' = (1 - p_k)M$. Then either the next
measurement is from a $\measure$ or $\cleanup$ operation. Suppose it is from
a $\measure$ operation. Then \[
\ket{\nu_{i,j,b}^a} = \frac{\ket{x} + (-1)^b \ket{y}}{\sqrt{2}}
\]
for two states $\ket{x}$, $\ket{y}$ that may or may not be in the $M'$-element
superposition held by the algorithm. Call this superposition $\ket{\phi}$. Then \[
\abs{\braket{\phi}{\nu_{i,j,b}^a}}^2 = \begin{cases}
\frac{(1 + (-1)^b)^2}{2M'} = \frac{(1 + (-1)^b)^2}{2(1 - p_k)M} & \mbox{if both
states are in the superposition}\\
\frac{1}{2M'} = \frac{1}{2(1 - p_k)M} & \mbox{if one of them is}\\
0 & \mbox{otherwise.}
\end{cases}
\]
So in the first of these cases, the probability that one of
$\ket{\nu_{i,j,b}^a}_{b = 1, 2}$ is returned, terminating the algorithm, is
$2/M'$. So $1 - p_{k+1} = (1 - 2/(1 - p_k)M)(1 - p_k) = 1 - p_k - 2/M$. If the
algorithm does not terminate, both states are removed from the superposition
and $M'$ becomes $(1-p_k)M - 2 = (1 - p_{k+1})M$.

In the second, the probability of termination is $2 \cdot 1/2(1 - p_k)M = 1/(1 -
p_k)M$, and so $1 - p_{k+1} = (1 - 1/(1 - p_k)M)(1 - p_k) = 1 - p_k - 1/M$. If
the algorithm does not terminate, one state is removed from the superposition
and $M'$ becomes $(1 - p_k)M - 1 = (1 - p_{k+1})M$.

In the third, $p_{k+1} = p_k$ and the superposition is unchanged, so $M'$
remains $(1 - p_k)M = (1 - p_{k+1})M$.

Now suppose the measurement is from a $\cleanup$ operation. Then the
measurement uses a projector onto a state $\ket{x}$ that may or may not be in
the superposition. If it is, the probability of termination is $1/(1 - p_k)M$
and $M'$ becomes $(1 - p_k)M - 1$ if the algorithm does not terminate, so $M' =
(1 - p_{k+1})M$. If it is not, $p_{k+1} = p_k$ and $M'$ is unchanged, so $M'$
remains $(1-p_{k+1})M$.
\end{proof}
\begin{lemma}
\label{lm:rightdegreecontribution}
For every $\dedge{uv} \in E$ such that $\db{\dedge{uv}}_u \in
\interval{d_{\alpha}, d_{\alpha+1}}$, $\db{\dedge{uv}}_v \in
\interval{d_{\beta}, d_{\beta+1}}$, and for every $(i,j) \in
\brac{\accuracy}^2$, the total expected contribution of $\measure(u,v)$ to the
pseudosnapshot estimate from returning one of
$\ket{\nu_{i,j,b}^k}_{k\in\brac{4},b \in \bool}$ is $1$ to the
$\bps{\dedge{uv}}_u$, $\bps{\dedge{uv}}_v$ entry and zero to all other entries
if $i = \frac{\accuracy}{2d_{\alpha}} \doutbps{\dedge{uv}}_u + 1$ and $j =
\frac{\accuracy}{2d_{\beta}} \doutbps{\dedge{uv}}_v + 1$.  Otherwise it is zero
everywhere.
\end{lemma}
\begin{proof}
We will start by analyzing the expectation conditional on the algorithm not
terminating before $\dedge{uv}$ arrives. This will give us the result in terms
of $M'$. We will then use Lemma~\ref{lm:earlyterm} to give us the expectation
in terms of $M$.

At the time when $\dedge{uv}$ arrives, let $R_u'$, $R_v'$, be the $R$ in
Lemma~\ref{lm:stinv} for $u,v$ respectively, and likewise for $r_u'$, $r_v'$
and $r$. Then, for $w = u,v$, set $r_w = r_w' + 1$. Set $R_u = R_u' +
f_\alpha(\dedge{uv}) + f_\beta(\dedge{uv})$. Note that this means that, before
measuring, the algorithm will update the state with an $\inc(\Ec, w, 1)$
operation for $w = u,v$ and $\Ec = \Ac, \Bc, \Cc, \Dc$, 
an $\inc(\Ec,u,d_{\alpha+1}\cdot f_{\alpha}(\dedge{uv}))$ operation for $\Ec = \Ac, \Bc$
and an
$\inc(\Ec,u,d_{\beta+1}\cdot f_{\beta}(\dedge{uv}))$ operation for $\Ec = \Cc, \Dc$.

Note that $R_u = \frac{\accuracy}{2d_{\alpha}}\doutbps{\dedge{uv}}_u$, $R_v =
\frac{\accuracy}{2d_{\beta}} \doutbps{\dedge{uv}}_v$, $r_u = \db{\dedge{uv}}_u$, $r_v =
\db{\dedge{uv}}_v$. So we have $r_u \in \interval{d_\alpha,d_{\alpha+1}}$, $r_v
\in \interval{d_\beta,d_{\beta+1}}$.

By Lemma~\ref{lm:stinv}, this means that for all $i$, before the increments,
\begin{align*}
d_\alpha + id_{\alpha+1} - 1 \in A_u, id_{\alpha+1} - 1 \in B_u &\mbox{ iff $i
\in \brac{R_u'}$}\\
d_\beta + id_{\beta+1} - 1 \in C_v, id_{\beta+1} - 1 \in D_v &\mbox{ iff $i \in
\brac{R_v'}$}
\end{align*}
and after them:
\begin{align*}
d_\alpha + id_{\alpha+1} \in A_u, id_{\alpha+1} \in B_u &\mbox{ iff $i
\in \brac{R_u}$}\\
d_\beta + id_{\beta+1} \in C_v, id_{\beta+1} \in D_v &\mbox{ iff $i \in
\brac{R_v}$}
\end{align*}
Now, writing $\ket{\phi}$ for the state of the algorithm before the
measurements, and recalling that
\begin{align*}
\ket{\nu_{i,j,b}^1} &= \frac{\ket{u(d_\alpha + (i-1)d_{\alpha+1})a^{r_{i,j}}} + (-1)^b
\ket{v(d_\beta + (j-1)d_{\beta+1})c^{r_{i,j}}}}{\sqrt{2}}\\
\ket{\nu_{i,j,b}^2} &= \frac{\ket{u(id_{\alpha+1})b^{r_{i,j}}} + (-1)^b \ket{v(d_\beta
+ (j-1)d_{\beta+1})c^{s_{i,j}}}}{\sqrt{2}}\\
\ket{\nu_{i,j,b}^3} &= \frac{\ket{u(d_\alpha + (i-1)d_{\alpha+1})a^{s_{i,j}}} + (-1)^b
\ket{v(jd_{\beta+1})d^{r_{i,j}}}}{\sqrt{2}}\\
\ket{\nu_{i,j,b}^4} &= \frac{\ket{u(id_{\alpha+1})b^{s_{i,j}}} + (-1)^b
\ket{v(jd_{\beta+1})d^{s_{i,j}}}}{\sqrt{2}}
\end{align*}
we have
\begin{align*}
\abs{\braket{\phi}{\nu_{i,j,b}^1}}^2 &= \begin{cases}
\frac{\paren*{1 + (-1)^b)}^2}{2M'}  & \mbox{$i \in \brac{R_u + 1}$ and $j \in \brac{R_v + 1}$}\\
\frac{1}{2M'} & \mbox{exactly one of $i \in \brac{R_u + 1}$ and $j \in \brac{R_v + 1}$}\\
0 & \mbox{otherwise}
\end{cases}\\
\abs{\braket{\phi}{\nu_{i,j,b}^2}}^2 &= \begin{cases}
\frac{\paren*{1 + (-1)^b)}^2}{2M'}  & \mbox{$i \in \brac{R_u}$ and $j \in \brac{R_v + 1}$}\\
\frac{1}{2M'} & \mbox{exactly one of $i \in \brac{R_u}$ and $j \in \brac{R_v + 1}$}\\
0 & \mbox{otherwise}
\end{cases}\\
\abs{\braket{\phi}{\nu_{i,j,b}^3}}^2 &= \begin{cases}
\frac{\paren*{1 + (-1)^b)}^2}{2M'}  & \mbox{$i \in \brac{R_u + 1}$ and $j \in \brac{R_v}$}\\
\frac{1}{2M'} & \mbox{exactly one of $i \in \brac{R_u + 1}$ and $j \in \brac{R_v}$}\\
0 & \mbox{otherwise}
\end{cases}\\
\abs{\braket{\phi}{\nu_{i,j,b}^4}}^2 &= \begin{cases}
\frac{\paren*{1 + (-1)^b)}^2}{2M'}  & \mbox{$i \in \brac{R_u}$ and $j \in \brac{R_v}$}\\
\frac{1}{2M'} & \mbox{exactly one of $i \in \brac{R_u}$ and $j \in \brac{R_v}$}\\
0 & \mbox{otherwise.}
\end{cases}
\end{align*}
Now, recall that the contribution to the chosen entry of the pseudosnapshot
(with the choice depending only on $i,j$) when seeing the result
$\ket{\nu_{i,j,b}^a}$  is $(-1)^bM/2$ if $a = 1,4$ and $(-1)^{\overline{b}}M/2$
if $a = 2,3$. Therefore, writing the total expected contribution from
$\ket{\nu_{i,j,b}^a}$ summed over $b = 0,1$ as $x_{i,j}^a$, we have
\begin{align*}
x_{i,j}^1 &= \begin{cases}
\frac{M}{M'} & \mbox{$i \in \brac{R_u + 1}$ and $j \in \brac{R_v + 1}$}\\
0 & \mbox{otherwise}
\end{cases}\\
x_{i,j}^2 &= \begin{cases}
-\frac{M}{M'} & \mbox{$i \in \brac{R_u}$ and $j \in \brac{R_v + 1}$}\\
0 & \mbox{otherwise}
\end{cases}\\
x_{i,j}^3 &= \begin{cases}
-\frac{M}{M'} & \mbox{$i \in \brac{R_u + 1}$ and $j \in \brac{R_v}$}\\
0 & \mbox{otherwise}
\end{cases}\\
x_{i,j}^4&= \begin{cases}
\frac{M}{M'} & \mbox{$i \in \brac{R_u}$ and $j \in \brac{R_v}$}\\
0 & \mbox{otherwise}
\end{cases}
\end{align*}
and so \[
\sum_{a=1}^4x_{i,j}^a = \begin{cases}
\frac{M}{M'} & \mbox{$i = R_u + 1$ and $j = R_v + 1$}\\
0 & \mbox{otherwise.}
\end{cases}
\]
Now, when $i = R_u + 1$ and $j = R_v + 1$, and $\db{\dedge{uv}}_u \in
\interval{d_{\alpha}, d_{\alpha+1}}$, $\db{\dedge{uv}}_v \in
\interval{d_{\beta}, d_{\beta+1}}$, the entry of the pseudosnapshot estimate
chosen is $\bps{\dedge{uv}}_u, \bps{\dedge{uv}}_v$. We have that the
contribution to it conditioned on the algorithm not terminating before
processing $\dedge{uv}$ is $M/M'$. So, as the contribution from $\dedge{uv}$ is
guaranteed to be $0$ if the algorithm has already terminated, the total
expectation is, by Lemma~\ref{lm:earlyterm}, \[ 
(1 - p) \cdot \frac{M}{M'} = (1 - p) \cdot \frac{M}{(1-p)M} = 1
\]
where $p$ is the probability of termination before processing $\dedge{uv}$.
\end{proof}

\begin{lemma}
\label{lm:wrongdegreecontribution}
For every $\dedge{uv}$ such that $\db{\dedge{uv}}_u \not\in
\interval{d_{\alpha}, d_{\alpha+1}}$, or $\db{\dedge{uv}}_v \not\in
\interval{d_{\beta}, d_{\beta+1}}$, the total expected contribution of
$\measure(u,v)$ to any entry of the pseudosnapshot estimate is zero.
\end{lemma}
\begin{proof}
We will analyze the expectation conditional on the algorithm not
terminating before $\dedge{uv}$ arrives. As the expected contribution is
trivially $0$ conditional on the algorithm terminating before $\dedge{uv}$
arrives, this will suffice for the result. 

We have that at least one of $\db{\dedge{uv}}_w < d_\gamma$ for some
$(w,\gamma) \in \set{(u,\alpha), (v,\beta)}$ or $\db{\dedge{uv}}_w \ge
d_{\gamma + 1}$ for some $(w,\gamma) \in \set{(u,\alpha), (v,\beta)}$. By
Lemma~\ref{lm:stinv}, this implies that, for one of $(E,F) \in \set{(A,B),
(C,D)}$, \[
d_\gamma + (i - 1)d_{\gamma+1} \in E_w \Leftrightarrow id_{\gamma+1} \in F_w
\]
for all $i \in \brac{M}$. So, writing $\ket{\phi}$ for the state
of the algorithm, we have either \[
\abs{\braket{\phi}{\nu_{i,j,b}^1}} = \abs{\braket{\phi}{\nu_{i,j,b}^2}},
\abs{\braket{\phi}{\nu_{i,j,b}^3}} = \abs{\braket{\phi}{\nu_{i,j,b}^4}}
\]
for all $i,j,b$, or \[
\abs{\braket{\phi}{\nu_{i,j,b}^1}} = \abs{\braket{\phi}{\nu_{i,j,b}^3}},
\abs{\braket{\phi}{\nu_{i,j,b}^2}} = \abs{\braket{\phi}{\nu_{i,j,b}^4}}
\]
for all $i,j,b$. As for all $i,j,b$ the contribution from the measurement
result $\ket{\nu_{i,j,b}^a}$ is made to the same entry of the estimate
regardless of $a$, and is $(-1)^bM/2$ for $a = 1,4$ and $(-1)^{\overline{b}}M/2$
for $a = 2,3$, this implies the expected contributions cancel out, and so the
lemma follows.
\end{proof}

\begin{lemma}
\label{lm:bias}
Each entry of the estimate has bias at most the number of edges
$\dedge{uv}$ such that:
\begin{enumerate}
\item $\db{\dedge{uv}}_u \in \interval{d_\alpha,
d_{\alpha+1}}$
\item $\db{\dedge{uv}}_v \in \interval{d_\beta, d_{\beta+1}}$
\item $\max\set*{\frac{\accuracy}{2d_{\alpha}} \doutbps{\dedge{uv}}_u + 1,
\frac{\accuracy}{2d_{\beta}} \doutbps{\dedge{uv}}_v + 1} > \accuracy$
\end{enumerate}
\end{lemma}
\begin{proof}
By Lemma~\ref{lm:wrongdegreecontribution}, no edge edge with endpoints in the
wrong degree classes contribute to the estimate. By
Lemma~\ref{lm:rightdegreecontribution} every edge $\dedge{uv}$ with endpoints
in the right degree classes contributes $1$ to the correct entry of the
estimate of $\histps{G}$ (and 0 to all others), unless \[
\max\set*{\frac{\accuracy}{2d_{\alpha}} \doutbps{\dedge{uv}}_u + 1,
\frac{\accuracy}{2d_{\beta}} \doutbps{\dedge{uv}}_v + 1} > \accuracy
\]
in which case it contributes 0 to all entries, as the $i,j$ in
Lemma~\ref{lm:rightdegreecontribution} only range in $\brac{\accuracy}^2$.
\end{proof}

\begin{lemma}
The variance of any entry of the pseudosnapshot estimate is $\bO{\accuracy^6m^2}$.
\end{lemma}
\begin{proof}
The output of the algorithm is an estimate of the pseudosnapshot with one
$\bO{M}$ entry and all other entries $0$. So this follows by the fact that $M =
\bO{\accuracy^3 m}$.
\end{proof}

\subsubsection{Space Usage}
\begin{lemma}
The algorithm uses only $\bO{\log n}$ qubits of space.
\end{lemma}
\begin{proof}
The algorithm maintains a single state that requires $\bO{\log M} = \bO{\log
n}$ qubits to store, along with a constant number of counters of size
$\poly(n)$, along with some rational numbers made from constant-degree
polynomials of such numbers.
\end{proof}

%% file: max_dicut_quantum_algorithm.tex
\section{Quantum Algorithm for \mdcut}\label{section:max_dicut_quantum_algorithm}
We may now prove our main result.
\begin{restatable}{theorem}{mdcutalg}
\label{thm:mdcutalg}
There is a quantum streaming algorithm which $0.4844$-approximates the \mdcut{}
value of an input graph $G$ with probability $1 - \delta$. The algorithm uses
$\bO{\log^5 n \log \frac{1}{\delta}}$ qubits of space.
\end{restatable}
\begin{proof}
Now, by Lemma~\ref{lm:dicuttopseudobias}, there is an $\alpha > 0.4844$, 
vector $\rb$ and constant $C > 0$ such that \[
\rb^\dagger \histps{G}\paren{1 - \rb} - C \eps m
\]
gives an $\alpha' = \alpha - \bO{\eps}$-approximation of the $\mdcut$ value of $G$ with
probability $1 - \bO{\eps} - e^{\bO{\eps^6 \kappa}}$. 

Recall that $\histps{G}$ depends on the accuracy parameters $\accuracy$ and
$\eps$ we choose. We will choose $\eps$ to be a small enough constant that
$\alpha' > 0.4844$ and the $1 - \bO{\eps}$ term in the success probability is at least
$4/5$. $\kappa$ will be chosen to depend only on $\eps$. We will choose it such
that the $1 - \bO{\eps} - e^{\bO{\eps^6 \kappa}}$ success probability is at
least $3/4$.

Now we just need to approximate $\histps{G}$. By applying
Lemma~\ref{lm:dicut-sketch} $\bO{\log^2 n}$ times in parallel (once for each
pair of degree classes $\interval{d_i,d_{i+1}}$) and summing the outputs, we
get a $\bO{\log^3 n}$ space (as $\eps$ and therefore $\kappa$ will be chosen to
be constant) algorithm that outputs an estimate of $\histps{G}$ where every
entry has variance $\bO{\kappa^6 m^2\log^2n}$, and bias
at most the number of edges $\dedge{uv}$ such that \[
\max\set*{\frac{\accuracy}{2d_{\alpha}} \doutbps{\dedge{uv}}_u + 1,
\frac{\accuracy}{2d_{\beta}} \doutbps{\dedge{uv}}_v + 1} > \accuracy
\]
where $\alpha$, $\beta$ are the unique indices in $\set{0,\dots,\floor{\log_{1
+ \eps^3} n}}$ such that $u \in \interval{d_\alpha, d_{\alpha+1}}$ and $v \in
\interval{d_{\beta}, d_{\beta+1}}$. Now for each of $(w,\gamma) = (u,\alpha),
(v,\beta)$, $\frac{\accuracy}{2d_{\gamma}}\doutbps{\dedge{uv}}_w$ is
distributed as $\text{B}(\db{\dedge{uv}}_w, \accuracy/2d_{\gamma})$, and so as
$\db{\dedge{uv}}_w \le (1 + \eps^3)d_{\gamma}$, the probability that it is at
least $\accuracy - 1$ is $2^{-\bOm{\accuracy}}$. So the expectation of the
total bias is at most $2^{-\bOm{\accuracy}}m$.

So by Markov's inequality, this bias is at most $2^{-\bOm{\accuracy}} m$ with
probability at most $1 - \eps$ over $f$ if we choose $\accuracy$ (as a function
of $\eps$) small enough. 

We may repeat this process $\bT{\kappa^6\frac{1}{\eps^3}\log^2n}$ times in
parallel (with the same hash functions each time), at $\bO{\log^5 n}$ space
cost, and pointwise average the snapshot estimates in order to obtain a
$\bO{\eps^3 m}$ variance estimate of $\histps{G}$. By Chebyshev's inequality,
with probability $1 - \eps$, every entry of this estimate is within $\eps m$ of
the corresponding entry of $\histps{G}$. Therefore, calling this estimate $M$,
\[
\rb^\dagger M\paren{1 - \rb} - D\eps m
\]
gives an $(\alpha' - \bO{\eps})$-approximation of the $\mdcut$ value of $G$
with probability $3/4-\bO{\eps}$. So choosing $\eps$ to be a small enough
constant, this gives us a 0.4844-approximation with probability $2/3$.

The result then follows by running the entire procedure $\bT{\log
\frac{1}{\delta}}$ times in parallel (sampling new hash functions each time),
and taking the median of the $\mdcut\paren{G}$ estimates.
\end{proof}

%% file: acknowledgements.tex
\section*{Acknowledgements}

Nadezhda Voronova thanks Mark Bun for many helpful conversations. Her research
was supported by NSF award CNS-2046425.

John Kallaugher and Ojas Parekh were supported by the Laboratory Directed
Research and Development program at Sandia National Laboratories, a
multimission laboratory managed and operated by National Technology and
Engineering Solutions of Sandia, LLC., a wholly owned subsidiary of Honeywell
International, Inc., for the U.S. Department of Energy's National Nuclear
Security Administration under contract DE-NA-0003525. Also supported by the
U.S. Department of Energy, Office of Science, Office of Advanced Scientific
Computing Research, Accelerated Research in Quantum Computing program, Fundamental Algorithmic Research for Quantum Computing (FAR-QC).  Part of this research was performed while Ojas Parekh was visiting the Institute for Pure and Applied Mathematics (IPAM), which is supported by the National Science Foundation (Grant No. DMS-1925919).

This article has been authored by an employee of National Technology \& Engineering Solutions of Sandia, LLC under Contract No. DE-NA0003525 with the U.S. Department of Energy (DOE). The employee owns all right, title and interest in and to the article and is solely responsible for its contents. The United States Government retains and the publisher, by accepting the article for publication, acknowledges that the United States Government retains a non-exclusive, paid-up, irrevocable, world-wide license to publish or reproduce the published form of this article or allow others to do so, for United States Government purposes. The DOE will provide public access to these results of federally sponsored research in accordance with the DOE Public Access Plan \url{https://www.energy.gov/downloads/doe-public-access-plan}.